\documentclass{symmetry10}

\begin{document}

\FirstPageHeading{Boyko&Popovych}

\ShortArticleName{Potential conservation laws of linear evolution equations}

\ArticleName{Simplest potential conservation laws\\ of linear evolution equations}

\Author{Vyacheslav M.~BOYKO~$^\dag$ and Roman O.~POPOVYCH~$^{\dag\ddag}$}
\AuthorNameForHeading{V.M.~Boyko and R.O.~Popovych}
\AuthorNameForContents{BOYKO V.M.\ and POPOVYCH R.O.}
\ArticleNameForContents{Simplest potential conservation laws of linear evolution equations}

\Address{$^\dag$~Institute of Mathematics of NAS of Ukraine,\\
\hphantom{$^\dag$}~3 Tereshchenkivska Str., Kyiv-4, Ukraine}
\EmailD{boyko@imath.kiev.ua, rop@imath.kiev.ua}

\Address{$^\ddag$~Fakult\"at f\"ur Mathematik, Universit\"at Wien,\\
\hphantom{$^\ddag$}~Nordbergstra{\ss}e 15, A-1090 Wien, Austria}

\Abstract{Every simplest potential conservation law of any $(1+1)$-dimensional linear evolution equation of even order proves
induced by a local conservation law of the same equation.
This claim is true also for linear simplest potential conservation laws of $(1+1)$-dimensional linear evolution equations of odd order,
which are related to linear potential systems.
We also derive an effective criterion for checking whether a quadratic conservation law of a simplest linear potential system
is a purely potential conservation law of a $(1+1)$-dimensional linear evolution equation of odd order.}

\section{Introduction}

The notion of potential conservation laws arises as a natural generalization of the notion of local conservation laws of differential equations.
We call a \emph{potential conservation law} any local conservation law of a potential system constructed from a given system~$\mathcal S$
of differential equations via introducing potentials by using local conservation laws of~$\mathcal S$ \cite{Boyko&Popovych:Popovych&Ivanova2005}.
This term first appeared in~\cite{Boyko&Popovych:Bluman&Doran-Wu1995}.
Potential conservation laws may be trivial in the sense that they induced by local conservation laws of the initial system
\cite{Boyko&Popovych:Kunzinger&Popovych2008,Boyko&Popovych:Popovych&Ivanova2005}.
The idea of iterative introduction of potentials by using local conservation laws of a potential system obtained on the previous step
was first suggested in the famous paper \cite{Boyko&Popovych:Wahlquist&Estabrook1975}
and later formalized in the form of the notion of {\em universal Abelian covering}
of differential equations~\cite{Boyko&Popovych:Bocharov&Co1999,Boyko&Popovych:Marvan2004,Boyko&Popovych:Sergyeyev2000}.
Although potential conservation laws of differential equations are interesting and important objects
for study within the framework of symmetry analysis,
nontrivial and complete results on such conservation laws were obtained only for a few classes of differential equations.
See related reviews and references in~\cite{Boyko&Popovych:Bluman&Cheviakov&Anco2010,Boyko&Popovych:Kunzinger&Popovych2008,Boyko&Popovych:Popovych&Ivanova2005}.

As a generalization of similar results from \cite{Boyko&Popovych:Popovych&Ivanova2005} for the linear heat equation,
it was proved in \cite{Boyko&Popovych:Popovych&Kunzinger&Ivanova2008} that
all potential conservation laws of $(1+1)$-dimensional linear second-order evolution equations are trivial.
Local conservation laws of these equations are well understood.
More precisely their spaces of the local conservation laws consist of linear conservation laws
the characteristics of which depend only upon independent variables and are solutions of the corresponding adjoint equations.
This finally solves the problem on potential conservation laws for these equations.

In the present paper we extend results from~\cite{Boyko&Popovych:Popovych&Kunzinger&Ivanova2008}
onto simplest potential conservation laws (i.e. those which involve only single potentials)
of $(1+1)$-dimensional linear second-order evolution equations to the case of an arbitrary order.
The extension to equations of even order is quite direct and exhaustive.
In the case of odd order we restrict our consideration with simplest potential systems
constructed by linear conservation laws.

Consider an arbitrary linear evolution equation of order~$n$ in the two independent variables $t$ and $x$
and the dependent variable~$u$,
\begin{gather}\label{Boyko&Popovych:EqGenLPE}
u_t=\sum_{i=0}^n A^iu_i,
\end{gather}
\looseness=-1
where $A^i=A^i(t,x)$ are arbitrary smooth functions,
$A^n\ne0$, $n\in\{2,3,4,\dots\}$,
$u_t=\partial u/\partial t$, $u_i\equiv \partial^i u/\partial x^i$, $i=1,\dots,n$, $u_0\equiv u$
(i.e., the function~$u$ is assumed to be its zeroth-order derivative).
We also employ, depending upon convenience or necessity, the following notation: $u_x=u_1$, $u_{xx}=u_2$ and $u_{xxx}=u_3$.
$D_t$ and $D_x$ are the operators of total differentiation with respect to the variables~$t$ and~$x$, respectively,
and $\mathop{\rm Div}$ denotes the total divergence,
$\mathop{\rm Div}\mathcal V=D_tF+D_xG$ for the tuple $\mathcal V=(F,G)$ of differential functions~$F$ and~$G$.
See~\cite{Boyko&Popovych:Olver1993,Boyko&Popovych:Popovych&Kunzinger&Ivanova2008} for other related notions.

Our paper is organized as follows:
In the next section
we briefly overview results of~\cite{Boyko&Popovych:Popovych&Sergyeyev2010} on local conservation laws of equations from class~\eqref{Boyko&Popovych:EqGenLPE}
and extend the simplest potential frame, constructed in~\cite{Boyko&Popovych:Popovych&Kunzinger&Ivanova2008}
for $(1+1)$-dimensional linear evolution equations of order two, to the case of an arbitrary order.
The same extension of dual Darboux transformations is carried out in Section~\ref{Boyko&Popovych:SectionOnDualDarbouxTransformation}.
Simplest potential conservation laws of $(1+1)$-dimensional linear evolution equations of even and odd order
are studied in Sections~\ref{Boyko&Popovych:SectionOnSimplestPotCLsofLPEsOfEvenOrder} and~\ref{Boyko&Popovych:SectionOnSimplestPotCLsofLPEsOfOddOrder}, respectively.
In the conclusion we discuss the results obtained and related problems which are still open.

\section{Local conservation laws and simplest\\ potential frame}\label{Boyko&Popovych:SectionOnLocalCLsAndSimplestPoteFrame}

It is well known that any linear partial differential equation~$\mathcal L$ admits cosymmetries
which are functions of independent variables only and are solutions of the corresponding adjoint equation~$\mathcal L^\dag$,
and every solution of~$\mathcal L^\dag$ is a cosymmetry.
Moreover any such cosymmetry is a characteristic of a conservation law of~$\mathcal L$
which contains a conserved vector linear in the dependent variable and its derivatives, cf.~\cite{Boyko&Popovych:Bluman&Shtelen2004}.
Following \cite[Section~5.3]{Boyko&Popovych:Olver1993} we call the conservation laws of this kind \emph{linear}.

It turns out that for any linear $(1+1)$-dimensional evolution equation of even order
its space of conservation laws is exhausted by linear ones
and therefore is isomorphic to the solution space of the corresponding adjoint equation \cite{Boyko&Popovych:Popovych&Sergyeyev2010}.
In~other words any cosymmetry of~\eqref{Boyko&Popovych:EqGenLPE} does not depend upon  derivatives of~$u$,
and a function $\alpha=\alpha(t,x)$ is a cosymmetry of~\eqref{Boyko&Popovych:EqGenLPE} if and only if it is a solution of the adjoint equation
\begin{gather}\label{Boyko&Popovych:EqAdjLPE}
\alpha_t+\sum_{i=0}^n (-1)^i (A^i\alpha)_i=0.
\end{gather}
Any such cosymmetry is a characteristic of a linear conservation law for~\eqref{Boyko&Popovych:EqGenLPE} with
the canonical conserved vector $\mathcal V=(F,G)$, where
\begin{gather}\label{Boyko&Popovych:EqCanonicalCVsOfLinCLsOfGenLinEvolEq}
F=\alpha u,\quad G=\sum_{i=0}^{n-1}\sigma^iu_i
\end{gather}
and the coefficients $\sigma^i=\sigma^i(t,x)$ are found recursively from the relations
\begin{gather}\label{Boyko&Popovych:sigma-lin}
\sigma^{n-1}=-\alpha A^n, \quad \sigma^i=-\alpha A^{i+1}-\sigma^{i+1}_x, \quad
i=n-2,\dots,0.
\end{gather}

For any linear $(1+1)$-dimensional evolution equation of odd order
the space of its conservation laws is spanned by linear and quadratic conservation laws~\cite{Boyko&Popovych:Popovych&Sergyeyev2010}.
There exist both such equations possessing infinite series of quadratic conservation laws of arbitrarily high orders
and ones having no quadratic conservation laws.
For all the formulas and claims obtained for equations of even order to be correct in the case of odd order
it is necessary to restrict the consideration with
linear local conservation laws, associated (linear) potential systems and their linear conservation laws.

Following the presentation of Section~7 of~\cite{Boyko&Popovych:Popovych&Kunzinger&Ivanova2008}
we investigate certain objects related to simplest potential systems of~\eqref{Boyko&Popovych:EqGenLPE}, i.e.
potential systems associated with single local conservation laws~\cite{Boyko&Popovych:Popovych&Ivanova2005}.
The theory of Darboux transformations for linear evolution equations~\cite{Boyko&Popovych:Matveev&Salle1991}
is strongly employed for this.
A detailed study of the simplest potential systems is necessary for understanding the general case since
such systems are components of more general potential systems.

Introducing the \emph{potential}~$v$ by the nontrivial canonical conserved vector~\eqref{Boyko&Popovych:EqCanonicalCVsOfLinCLsOfGenLinEvolEq}
associated with the characteristic~$\alpha=\alpha(t,x)\ne0$,
we obtain the potential system
\begin{equation}\label{Boyko&Popovych:EqPotSysOfLPE}
v_x=\alpha u,\quad v_t=-\sum_{i=0}^{n-1}\sigma^iu_i.
\end{equation}

The initial equation~\eqref{Boyko&Popovych:EqGenLPE} for $u$ is a differential consequence of system~\eqref{Boyko&Popovych:EqPotSysOfLPE}.
Another differential consequence of~\eqref{Boyko&Popovych:EqPotSysOfLPE} is the equation
\begin{equation}\label{Boyko&Popovych:EqPotLPE}
v_t=- \sum_{i=0}^{n-1}\sigma^i \Big(\frac{v_x}{\alpha}\Big)_i
\end{equation}
\looseness=-1
on the potential dependent variable~$v$, which is called the \emph{potential equation}
associated with the linear evolution equation~\eqref{Boyko&Popovych:EqGenLPE} and the characteristic~$\alpha$.
There is a one-to-one correspondence between solutions of the potential system and the potential equation
due to the projection $(u,v)\to v$ on the one hand and due to the formula $u=v_x/\alpha$ on the other,
cf.\ \cite{Boyko&Popovych:Popovych&Kunzinger&Ivanova2008}.
The correspondence between solutions of the initial equation and the potential system is one-to-one only up to a constant summand.

It is convenient to use another dependent variable $w=\psi v$ instead of~$v$ in our further considerations,
where we introduce the notation $\psi=1/\alpha$.
The function~$w$ is called the \emph{modified potential} associated with the characteristic~$\alpha=1/\psi$.
After being written in terms of $w$ and $\psi$ instead of~$v$ and~$\alpha$,
the potential system~\eqref{Boyko&Popovych:EqPotSysOfLPE} and the potential equation~\eqref{Boyko&Popovych:EqPotLPE} take the form
\begin{equation}\label{Boyko&Popovych:EqModifiedPotSysOfLPE}
w_x-\frac{\psi_x}\psi w=u,\quad w_t-\frac{\psi_t}\psi w=-\psi\sum_{i=0}^{n-1}\sigma^iu_i
\end{equation}
and
\begin{equation}\label{Boyko&Popovych:EqModifiedPotLPE}
w_t=-\psi\sum_{i=0}^{n-1}\sigma^i\biggl(w_x-\frac{\psi_x}\psi w\biggr)_i+\frac{\psi_t}\psi w=:\sum_{i=0}^nB^iw_i,
\end{equation}
respectively. Here
\begin{gather*}
B^i=-\psi\sigma^{i-1}+\psi\sum_{j=i}^{n-1}\binom ij\,\sigma^i\biggl(\frac{\psi_x}\psi\biggr)_{j-i},\quad
i=n,n-1,\dots,1,\\
B^0=\frac{\psi_t}\psi+\psi\sum_{j=0}^{n-1}\sigma^i\biggl(\frac{\psi_x}\psi\biggr)_j.
\end{gather*}
In particular $B^n=A^n$ and $B^{n-1}=A^{n-1}-A^n_x$.
System~\eqref{Boyko&Popovych:EqModifiedPotSysOfLPE} and equation~\eqref{Boyko&Popovych:EqModifiedPotLPE}
are called the \emph{modified potential system} and
the \emph{modified potential equation} associated with the characteristic~$\alpha$.
These representations of the potential system and potential equation are more suitable for
the study within the framework of symmetry analysis.

As the function $v=1$ obviously is a solution of~\eqref{Boyko&Popovych:EqPotLPE} and therefore
the function $w=\psi$ is a solution of~\eqref{Boyko&Popovych:EqModifiedPotLPE},
the first equation of~\eqref{Boyko&Popovych:EqModifiedPotSysOfLPE} in fact represents the Darboux transformation~\cite{Boyko&Popovych:Crum1955,Boyko&Popovych:Matveev&Salle1991}
of~\eqref{Boyko&Popovych:EqModifiedPotLPE} to~\eqref{Boyko&Popovych:EqGenLPE}.

\section{Dual Darboux transformation}\label{Boyko&Popovych:SectionOnDualDarbouxTransformation}

The remarkable fact that Darboux covariance holds for $(1+1)$-dimensional linear evolution equations of arbitrary order
was first established in~\cite{Boyko&Popovych:Matveev1979} (see also~\cite[p.~17]{Boyko&Popovych:Matveev&Salle1991}).
In contrast to the previous section for the coherent presentation we assume below that
the initial object of the consideration is the equation~\eqref{Boyko&Popovych:EqModifiedPotLPE},
\begin{equation}\label{Boyko&Popovych:EqModifiedPotLPEUnconstrained}
w_t=\sum_{i=0}^nB^iw_i,
\end{equation}
which is interpreted as an arbitrary representative of the class of linear evolution equations.

Denote by ${\rm DT}[\varphi]$ the Darboux transformation constructed with a nonzero solution~$\varphi$
of~\eqref{Boyko&Popovych:EqModifiedPotLPEUnconstrained}, i.e.,
\[
{\rm DT}[\varphi](w)=w_x-\frac{\varphi_x}\varphi w.
\]
The Darboux transformation possesses the useful property of duality.
We formulate this in the same way as in~\cite{Boyko&Popovych:Popovych&Kunzinger&Ivanova2008},
which is slightly different from~\cite[Section~2.4]{Boyko&Popovych:Matveev&Salle1991}.

\begin{lemma}\label{Boyko&Popovych:LemmaOnDualDarbouxTrans}
Let $w^0$ be a fixed nonzero solution of~\eqref{Boyko&Popovych:EqModifiedPotLPEUnconstrained} and let
the Darboux transformation ${\rm DT}[w^0]$ map~\eqref{Boyko&Popovych:EqModifiedPotLPEUnconstrained} to equation~\eqref{Boyko&Popovych:EqGenLPE}.
Then $\alpha^0=1/w^0$ is a solution of the equation~\eqref{Boyko&Popovych:EqAdjLPE} adjoint to equation~\eqref{Boyko&Popovych:EqGenLPE}
and ${\rm DT}[\alpha^0]$ maps~\eqref{Boyko&Popovych:EqAdjLPE} to the equation adjoint to~\eqref{Boyko&Popovych:EqModifiedPotLPEUnconstrained}, i.e.,
$$\begin{array}{rcl}
\displaystyle \smash{u_t=\sum_{i=0}^n A^iu_i} & \xleftarrow{{\rm DT}[w^0]} & \displaystyle \smash{w_t=\sum_{i=0}^nB^iw_i},
\\[1.8ex]
&\Updownarrow&
\\[1.8ex]
\displaystyle {\alpha_t+\sum_{i=0}^n (-1)^i (A^i\alpha)_i=0} & \xrightarrow{{\rm DT}[\alpha^0]} &
\displaystyle {\beta_t+\sum_{i=0}^n (-1)^i (B^i\beta)_i=0}.
\end{array}
$$
\end{lemma}

\begin{remark}
Similarly to~\cite{Boyko&Popovych:Popovych&Kunzinger&Ivanova2008}
the Darboux transformation ${\rm DT}[\alpha^0]$ is called the \emph{dual} to the Darboux transformation ${\rm DT}[w^0]$.
Since the twice adjoint equation coincides with the initial equation,
the twice dual Darboux transformation is nothing but the initial Darboux transformation.
This implies that `then' in Lemma~\ref{Boyko&Popovych:LemmaOnDualDarbouxTrans} can be replaced by `if and only if'.
\end{remark}

\begin{remark}\label{Boyko&Popovych:RemarkOnDarbouxTransAsLinearMappingOfSolutionSpaceOfLPEs}
The Darboux transformation ${\rm DT}[w^0]$ from Lemma~\ref{Boyko&Popovych:LemmaOnDualDarbouxTrans}
is a linear mapping of the solution space of~\eqref{Boyko&Popovych:EqModifiedPotLPEUnconstrained} to the solution space of~\eqref{Boyko&Popovych:EqGenLPE}.
The kernel of this mapping coincides with the linear span~$\langle w^0\rangle$.
Its image is the whole solution space of~\eqref{Boyko&Popovych:EqGenLPE}. Indeed for any solution~$u$ of~\eqref{Boyko&Popovych:EqGenLPE}
we can find a solution~$w$ of~\eqref{Boyko&Popovych:EqModifiedPotLPEUnconstrained}, mapped to~$u$, by integrating
system~\eqref{Boyko&Popovych:EqModifiedPotSysOfLPE} with respect to~$w$.
By the Frobenius theorem system~\eqref{Boyko&Popovych:EqModifiedPotSysOfLPE} is compatible in view of equation~\eqref{Boyko&Popovych:EqGenLPE}.
Therefore ${\rm DT}[w^0]$ generates a one-to-one linear mapping between
the solution space of~\eqref{Boyko&Popovych:EqModifiedPotLPEUnconstrained}, factorized by~$\langle w^0\rangle$,
and the solution space of~\eqref{Boyko&Popovych:EqGenLPE}.
\end{remark}

In the case of even order~$n$ Lemma~\ref{Boyko&Popovych:LemmaOnDualDarbouxTrans} jointly with Remark~\ref{Boyko&Popovych:RemarkOnDarbouxTransAsLinearMappingOfSolutionSpaceOfLPEs}
can be reformulated in terms of characteristics of
conservation laws. Denote equations~\eqref{Boyko&Popovych:EqGenLPE} and~\eqref{Boyko&Popovych:EqModifiedPotLPEUnconstrained}, where $n\in2\mathbb N$, by~$\mathcal L$
and~$\widehat{\mathcal L}$ for convenience.

\begin{lemma}\label{Boyko&Popovych:LemmaOnDualDarbouxTransInTermsOfChars}
If $w^0$ is a nonzero solution of~$\widehat{\mathcal L}$ and ${\rm DT}[w^0](\widehat{\mathcal L}\,)=\mathcal L$,
then $\alpha^0=1/w^0$ is a characteristic of~$\mathcal L$ and
the Darboux transformation ${\rm DT}[\alpha^0]$ maps the characteristic space of~$\mathcal L$
onto the characteristic space of~$\widehat{\mathcal L}$.
\end{lemma}

\begin{proposition}\label{Boyko&Popovych:PropositionOnDarbouxTransAsSplittingOperator}
Let $w^0$ be a fixed nonzero solution of~\eqref{Boyko&Popovych:EqModifiedPotLPEUnconstrained} and let
the Darboux transformation ${\rm DT}[w^0]$ map~\eqref{Boyko&Popovych:EqModifiedPotLPEUnconstrained} to equation~\eqref{Boyko&Popovych:EqGenLPE}.
Then the operator associated with ${\rm DT}[w^0]$ is a splitting operator for
the pair of operators associated with equations~\eqref{Boyko&Popovych:EqGenLPE} and~\eqref{Boyko&Popovych:EqModifiedPotLPEUnconstrained}, i.e.,
\[
\biggl(\partial_t-\sum_{i=0}^n A^i\partial_x^i\biggr){\rm DT}[w^0]={\rm DT}[w^0]\biggl(\partial_t-\sum_{i=0}^nB^i\partial_x^i\biggr).
\]
\end{proposition}



\section{Simplest potential conservation laws: even order}\label{Boyko&Popovych:SectionOnSimplestPotCLsofLPEsOfEvenOrder}

If a potential system is constructed by introducing a potential~$v$ with a single local conservation law of~\eqref{Boyko&Popovych:EqGenLPE},
each of its local conservation laws is a \emph{simplest potential conservation law} of~\eqref{Boyko&Popovych:EqGenLPE}, cf.~\cite{Boyko&Popovych:Popovych&Ivanova2005}.
We say that a simplest potential conservation law $\bar{\mathcal F}$ of~\eqref{Boyko&Popovych:EqGenLPE}
is \emph{induced} by a local conservation law~$\mathcal F$ of~\eqref{Boyko&Popovych:EqGenLPE}
if $\bar{\mathcal F}$ contains a conserved vector which is the pullback of a conserved vector from~$\mathcal F$
with respect to the projection
\[
\varpi\colon J^\infty(t,x\,|\,u,v)\to J^\infty(t,x\,|\,u),
\]
where $J^\infty(t,x\,|\,u,v)$ (resp.\ $J^\infty(t,x\,|\,u)$) denotes the jet space with the independent variables~$t$ and~$x$ and
the dependent variables $u$ and $v$ (resp.\ the dependent variable~$u$).
In view of Proposition~2 from~\cite{Boyko&Popovych:Kunzinger&Popovych2008} this is equivalent to that
the conservation law $\bar{\mathcal F}$ contains a conserved vectors depending upon~$t$, $x$ and derivatives of~$u$.

\begin{theorem}\label{Boyko&Popovych:TheoremOnSimplestPotCLsOfEvenOrderLinEvolEq}
Every simplest potential conservation law of any $(1+1)$-dimensional linear evolution equation of even order
is induced by a local conservation law of the same equation.
\end{theorem}

\begin{proof}
Potentials~$v$ and~$\tilde v$ introduced with equivalent conserved vectors are connected by the transformation
$\tilde v=v+f[u]$, where $f[u]$ is a function of~$t$, $x$ and derivatives of~$u$.
This transformation preserves the property of inducing simplest potential conservation laws by local ones.
Therefore exhaustively to investigate simplest potential conservation laws of equations from the class~\eqref{Boyko&Popovych:EqGenLPE} with even~$n$
it is sufficient to study local conservation laws of potential systems of the form~\eqref{Boyko&Popovych:EqPotSysOfLPE}
associated with canonical conserved vectors of the form~\eqref{Boyko&Popovych:EqCanonicalCVsOfLinCLsOfGenLinEvolEq}.

We fix an equation from the class~\eqref{Boyko&Popovych:EqGenLPE} and its characteristic~$\alpha$
and consider the corresponding potential system~\eqref{Boyko&Popovych:EqPotSysOfLPE}.
As the usual potential~$v$ is connected with the modified potential~$w$ via a point transformation,
we can investigate conservation laws of the modified potential system~\eqref{Boyko&Popovych:EqModifiedPotSysOfLPE} instead of~\eqref{Boyko&Popovych:EqPotSysOfLPE}.
Up to equivalence of conserved vectors on the solution set of~\eqref{Boyko&Popovych:EqModifiedPotSysOfLPE},
we can exclude derivatives of~$u$ from any conserved vector of~\eqref{Boyko&Popovych:EqModifiedPotSysOfLPE}.
In other words each local conservation law~$\bar{\mathcal F}$ of the modified potential system~\eqref{Boyko&Popovych:EqModifiedPotSysOfLPE} contains
a conserved vector depending solely on~$t$, $x$ and derivatives of~$w$ and therefore
is induced by a local conservation law of the modified potential equation~\eqref{Boyko&Popovych:EqModifiedPotLPE}.

As equation~\eqref{Boyko&Popovych:EqModifiedPotLPE} also is a $(1+1)$-dimensional linear evolution equation of even order
as the initial equation, its space of conservation laws is exhausted by linear conservation laws,
cf.\ the discussion in the beginning of Section~\ref{Boyko&Popovych:SectionOnLocalCLsAndSimplestPoteFrame}.
An arbitrary characteristic~$\beta$ of~\eqref{Boyko&Popovych:EqModifiedPotLPE} depends only upon~$t$ and~$x$ and satisfies the equation adjoint to~\eqref{Boyko&Popovych:EqModifiedPotLPE}.
In view of Lemma~\ref{Boyko&Popovych:LemmaOnDualDarbouxTransInTermsOfChars}
there exists a characteristic~$\tilde\alpha$ of~\eqref{Boyko&Popovych:EqGenLPE} such that $\beta={\rm DT}[\alpha]\tilde\alpha$.

The conserved vectors~$\mathcal V^1$ and~$\mathcal V^2$ of the modified potential system~\eqref{Boyko&Popovych:EqModifiedPotSysOfLPE},
which are the pullbacks of the canonical conserved vectors of
the initial equation~\eqref{Boyko&Popovych:EqGenLPE} and the modified potential equation~\eqref{Boyko&Popovych:EqModifiedPotLPE},
associated with the characteristic~$\tilde\alpha$ and~$\beta$, respectively, are equivalent.
Indeed the sum of the densities of~$\mathcal V^2$ and~$\mathcal V^1$ is
\[
\beta w+\tilde\alpha u=\Bigl(\tilde\alpha_x-\frac{\alpha_x}{\alpha}\tilde\alpha\Bigr)w
+\tilde\alpha\Bigl(w_x+\frac{\alpha_x}{\alpha}w\Bigr)=D_x(\tilde\alpha w).
\]
Denote by~$\mathcal V^0$ the trivial conserved vector $(-D_x(\tilde\alpha w),D_t(\tilde\alpha w))$.
The conserved vector $\mathcal V^0+\mathcal V^1+\mathcal V^2$ of the system~\eqref{Boyko&Popovych:EqModifiedPotSysOfLPE} has zero density and therefore
is a trivial conserved vector.
(In fact the conserved vector $\mathcal V^0+\mathcal V^1+\mathcal V^2$ equals zero.)
This means that the conserved vectors~$-\mathcal V^1$ and~$\mathcal V^2$ are equivalent.

In summary we prove that any simplest potential conservation law of equation~\eqref{Boyko&Popovych:EqGenLPE}
contains a conserved vector which is the pullback of a local conserved vector of~\eqref{Boyko&Popovych:EqGenLPE}.
\end{proof}

\begin{remark}\label{Boyko&Popovych:RemarkOnUsingCriterionOnPotCLs}
The explicit construction of a local conserved vector which is equivalent to~$\mathcal V^2$
in the end of the proof can be replaced by arguments based on the criterion of induction of potential conservation laws by local ones,
cf.\ Proposition~8 from~\cite{Boyko&Popovych:Kunzinger&Popovych2008}.
Indeed  the canonical conserved vector of the modified potential equation~\eqref{Boyko&Popovych:EqModifiedPotLPE},
which is the trivial projection of~$\mathcal V^2$, is associated with the characteristic $\beta=\beta(t,x)$.
This is why
\begin{align*}
\mathop{\rm Div}\mathcal V^2&=\beta\biggl(w_t-\sum_{i=0}^nB^iw_i\biggr)\\
&=\beta\biggl(w_t-\frac{\psi_t}\psi w+\psi\sum_{i=0}^{n-1}\sigma^iu_i+\psi\sum_{i=0}^{n-1}\sigma^iD_x^i\biggl(w_x-\frac{\psi_x}\psi w-u\biggr)\biggr)\\
&=\psi\beta\biggl(v_t+\sum_{i=0}^{n-1}\sigma^iu_i\biggr)+\psi\biggl(\sum_{i=0}^{n-1}(-D_x)^i(\psi\sigma^i\beta)\biggr)(v_x-\alpha u)
+D_x\Phi
\end{align*}
for some differential function~$\Phi$ of~$u$ and~$v$ for which the precise expression is not essential.
(It is a linear function in derivatives of~$u$ and~$v$ with coefficients depending upon~$t$ and~$x$.)
In other words the conserved vector~$\mathcal V^2$ belongs to the conservation law~$\bar{\mathcal F}$ of the potential system~\eqref{Boyko&Popovych:EqPotSysOfLPE}
with the characteristic having the components
\[
\psi\beta \quad\mbox{and}\quad \psi\sum_{i=0}^{n-1}(-D_x)^i(\psi\sigma^i\beta).
\]
This characteristic is completely reduced since it does not depend upon derivatives of~$u$ and~$v$.
In particular it does not depend upon~$v$.
In view of Proposition~8 from~\cite{Boyko&Popovych:Kunzinger&Popovych2008},
the associated local conservation law~$\bar{\mathcal F}$ of the potential system~\eqref{Boyko&Popovych:EqPotSysOfLPE}
is induced by a local conservation law of the initial equation~\eqref{Boyko&Popovych:EqGenLPE}.
\end{remark}

\section{Simplest potential conservation laws: odd order}\label{Boyko&Popovych:SectionOnSimplestPotCLsofLPEsOfOddOrder}

Theorem~\ref{Boyko&Popovych:TheoremOnSimplestPotCLsOfEvenOrderLinEvolEq} cannot be directly extended to $(1+1)$-dimensional linear evolution equations of odd order
since such equations may possess, additionally to linear, quadratic conservation laws.
At the same time it is easy to see that a similar statement can be proved for the case of odd order after restricting to the completely linear case.

\begin{theorem}\label{Boyko&Popovych:TheoremOnSimplestLinearPotCLsOfOddOrderLinEvolEq}
Every linear simplest potential conservation law of any $(1+1)$-dimensional linear evolution equation of odd order,
which is related to a linear potential system,
is induced by a local conservation law of the same equation.
\end{theorem}

A $(1+1)$-dimensional linear evolution equation~$\mathcal L$ of odd order may additionally possesses two kinds of simplest potential conservation laws:
\begin{itemize}\itemsep=0ex
\item
conservation laws of potential systems constructed with conserved vectors of~$\mathcal L$ which nontrivially contain terms quadratic in derivatives of~$u$ and
\item
quadratic conservation laws of simplest linear potential systems.
\end{itemize}
Potential conservation laws of the first kind are difficult for investigation.
Thus, in contrast to the linear case, related potential systems usually have no analogues of potential equations.
It seems that an only possibility to study conservation laws of these systems
is the direct application of general methods discussed, e.g.,
in~\cite{Boyko&Popovych:Anco&Bluman2002a,Boyko&Popovych:Anco&Bluman2002b,
Boyko&Popovych:Bluman&Cheviakov&Anco2010,Boyko&Popovych:Bocharov&Co1999,Boyko&Popovych:Popovych&Ivanova2005,Boyko&Popovych:Wolf2002}.

Here we consider only potential conservation laws of the second kind.
There exists a simple criterion to check whether a potential conservation law of this kind is induced by a local conservation law
of the initial equation~\eqref{Boyko&Popovych:EqGenLPE}.

\begin{theorem}\label{Boyko&Popovych:TheoremOnSimplestQuadraticPotCLsOfOddOrderLinEvolEq}
Let $\alpha=\alpha(t,x)$ be a nonzero characteristic of a $(1+1)$-dimensional linear evolution equation of odd order~\eqref{Boyko&Popovych:EqGenLPE}
and
\[
\gamma=\Gamma w, \quad \Gamma=\sum_{k=0}^r g^k(t,x)D_x^k, \quad g^r\ne0,
\]
be a characteristic of the corresponding modified potential equation~\eqref{Boyko&Popovych:EqModifiedPotLPE}.
Then the conservation law of potential system~\eqref{Boyko&Popovych:EqPotSysOfLPE}, which is associated with~$\gamma$,
is induced by a local conservation law of~\eqref{Boyko&Popovych:EqGenLPE} if and only if
the solution $\psi=1/\alpha$ of~\eqref{Boyko&Popovych:EqModifiedPotLPE} belongs to the kernel of operator~$\Gamma$, i.e., $\Gamma\psi=0$.
\end{theorem}

\begin{proof}
Denote by~$\mathcal V$ the conserved vector of the potential system~\eqref{Boyko&Popovych:EqPotSysOfLPE},
which is obtained by the pullback of the canonical conserved vector of the modified potential equation~\eqref{Boyko&Popovych:EqModifiedPotLPE},
associated with the characteristic~$\gamma$, and the consequent transformation $v=\alpha w$.
Analogously to Remark~\ref{Boyko&Popovych:RemarkOnUsingCriterionOnPotCLs} we have
\begin{align*}
\mathop{\rm Div}\mathcal V={}&\gamma\biggl(w_t-\sum_{i=0}^nB^iw_i\biggr)\\
={}&(\Gamma w)\biggl(w_t-\frac{\psi_t}\psi w+\psi\sum_{i=0}^{n-1}\sigma^iu_i+\psi\sum_{i=0}^{n-1}\sigma^iD_x^i\biggl(w_x-\frac{\psi_x}\psi w-u\biggr)\biggr)\\
={}&\psi\Gamma(\psi v)\biggl(v_t+\sum_{i=0}^{n-1}\sigma^iu_i\biggr)+\psi\biggl(\,\sum_{i=0}^{n-1}(-D_x)^i\bigl(\psi\sigma^i\Gamma(\psi v)\bigr)\biggr)(v_x-\alpha u)\\
&{}+D_x\Phi
\end{align*}
for some differential function~$\Phi$ of~$u$ and~$v$ the precise expression for which again is not essential.
(It is a quadratic function in derivatives of~$u$ and~$v$ with coefficients depending on~$t$ and~$x$.)
This means that the conserved vector~$\mathcal V$ belongs to the conservation law of the potential system~\eqref{Boyko&Popovych:EqPotSysOfLPE}
with the characteristic~$\lambda$ having the components
\[
\psi\Gamma(\psi v) \quad\mbox{and}\quad \psi\biggl(\,\sum_{i=0}^{n-1}(-D_x)^i\bigl(\psi\sigma^i\Gamma(\psi v)\bigr)\biggr).
\]
For the characteristic~$\lambda$ to be completely reduced we have to exclude the derivatives $v_k$, $k=1,\dots,r+n-1$, from it
using the differential consequences of the equation $v_x=\alpha u$.
The reduced form of~$\lambda$ depends upon the potential~$v$ if and only if $\Gamma\psi=0$.
Therefore the statement to be proved follows from the criterion of induction of potential conservation laws by local conservation laws
formulated in~\cite{Boyko&Popovych:Kunzinger&Popovych2008} as Proposition~8.
\end{proof}

\begin{example}
We construct an example starting from the corresponding modified potential equation with the known space of quadratic conservation laws.
Consider the ``linear Korteweg--de Vries equation''
\begin{equation}\label{Boyko&Popovych:EqLinearKdV}
w_t=w_{xxx}
\end{equation}
which is identical with its adjoint.
It was proved in~\cite{Boyko&Popovych:Popovych&Sergyeyev2010} that
the space of its quadratic conservation laws is spanned by the conservation laws with the characteristics
$\Gamma_{ml}w$, where $\Gamma_{ml}=D_x^m(3tD_x^2+x)^lD_x^m$ and $l,m=0,1,2,\dots$.

As the solution~$\psi$ of the modified potential equation we choose the function $w=x$.
The Darboux transformation ${\rm DT}[x]$ maps equation~\eqref{Boyko&Popovych:EqLinearKdV} to the ``initial'' equation
\begin{equation}\label{Boyko&Popovych:EqImageOfLinearKdVwrtDTx}
u_t=u_{xxx}-\frac3{x^2}u_x+\frac3{x^3}u
\end{equation}
which also is identical with its adjoint.
The solution $\alpha$ of the equation~\eqref{Boyko&Popovych:EqImageOfLinearKdVwrtDTx}, dual to~$\psi$, is $u=1/\psi=1/x$.
The potential system constructed by the conservation law of~\eqref{Boyko&Popovych:EqImageOfLinearKdVwrtDTx} with the characteristic $\alpha=1/x$
has the form
\[
v_x=\frac ux, \quad v_t=\frac{u_{xx}}x+\frac{u_x}{x^2}-\frac u{x^3}.
\]

We have that $\Gamma_{ml}\psi=0$ if and only if $m\geqslant2$.
Therefore a complete set of independent simplest purely potential conservation laws of the equation~\eqref{Boyko&Popovych:EqImageOfLinearKdVwrtDTx},
which are obtained via introducing the potential with the characteristic $\alpha=1/x$, is
exhausted by the quadratic conservation laws constructed by the pullback of the conservations laws of
the corresponding modified potential equation~\eqref{Boyko&Popovych:EqLinearKdV}, having the characteristics
$\Gamma_{ml}w$, where $l=0,1,2,\dots$ and $m=0,1$.
\end{example}



\section{Conclusions}

\looseness=1
In this paper we have studied simplest potential conservation laws of $(1+1)$-dimen\-sional linear evolution equations,
which are constructed via introducing single potentials using linear conservation laws.
Such conservation laws are quite trivial in the case of equations of even order:
All simplest potential conservation laws of any $(1+1)$-dimensional linear evolution equation of even order
are induced by local conservation laws of the same equation and
its space of local conservation laws is exhausted by linear ones.
Similar results concerning equations of odd order are more complicated.
Although all simplest linear potential conservation laws of these equations are induced by local ones,
it is not the case for quadratic conservation laws.
We derive an effective criterion which allows us to check easily whether a quadratic conservation law
of a simplest modified potential equation leads to a purely potential conservation law for the initial equation.
It is true if and only if any characteristic of this conservation law does not vanish for the value $w=1/\alpha$
of the modified potential~$w$, where $\alpha$ is the characteristic of the linear conservation law of the initial equation used for introducing the potential.

A preliminary analysis shows that all the results obtained in this paper for simplest conservation laws
could be easily extended to the case of an arbitrary number of potentials introduced with linear conservation laws.
The first step of this investigation should be the construction of the whole linear potential frame
for the class of $(1+1)$-dimensional linear evolution equations of an arbitrary fixed order
as this was done for order two in~\cite{Boyko&Popovych:Popovych&Kunzinger&Ivanova2008}.
It is obvious that the linear potential frame coincides with the entire potential frame if the order of the equation is even.

The consideration of nonlinear potential systems constructed for equations of odd order with quadratic conservation laws
calls for the development of new tools which are different from those used for linear potential systems.

Another possible direction for further investigations close to the subject of this paper
is the description of potential symmetries of $(1+1)$-dimensional linear evolution equation of arbitrary order,
at least those associated with linear potential systems.
It seems to us that the approach which was developed in~\cite{Boyko&Popovych:Popovych&Kunzinger&Ivanova2008} for the case of equations of order two
and is based upon the construction of the extended potential frame and the reduction of the consideration to the study
of single modified potential equations also has to work for an arbitrary order.

\subsection*{Acknowledgements}

The research was supported by the project P20632 of the Austrian Science Fund.
V.M.B.\ gratefully acknowledges
the warm hospitality extended to him by the Department of Mathematics of the University of Vienna
during his visits in the course of preparation of the present paper.
R.O.P.\ sincerely thanks the University of Cyprus and the Cyprus Research Promotion Foundation (project number $\Pi$PO$\Sigma$E$\Lambda$KY$\Sigma$H/$\Pi$PONE/0308/01)
for support of his participation in the Fifth Workshop on Group Analysis of Differential Equations and Integrable Systems
and the wonderful hospitality during the workshop.

\LastPageEnding


\begin{thebibliography}{99}
\footnotesize

\bibitem{Boyko&Popovych:Anco&Bluman2002a}
Anco~S.C. and Bluman~G.,
Direct construction method for conservation laws of partial differential equations.~I.
Examples of conservation law classifications,
{\it Eur. J. Appl. Math.}, 2002, V.13, 545--566, arXiv:math-ph/0108023.

\bibitem{Boyko&Popovych:Anco&Bluman2002b}
Anco~S.C. and Bluman~G.,
Direct construction method for conservation laws of partial differential equations.~II. General treatment,
{\it Eur. J. Appl. Math.}, 2002, V.13, 567--585, arXiv:math-ph/0108024.

\bibitem{Boyko&Popovych:Bluman&Cheviakov&Anco2010}
Bluman G.W., Cheviakov A.F. and Anco  S.C.,
{\it Applications of Symmetry Methods to Partial Differential Equations},
Springer, New York, 2010.

\bibitem{Boyko&Popovych:Bluman&Doran-Wu1995}
Bluman G. and Doran-Wu P.,
The use of factors to discover potential systems or linearizations. Geometric and algebraic structures in differential equations,
{\it  Acta Appl. Math.}, 1995, V.41, 21--43.

\bibitem{Boyko&Popovych:Bluman&Shtelen2004}
Bluman G. and Shtelen V., Nonlocal transformations of Kolmogorov equations into the backward heat equation,
{\it J. Math. Anal. Appl.}, 2004, V.291, 419--437.

\bibitem{Boyko&Popovych:Bocharov&Co1999}
Bocharov~A.V., Chetverikov~V.N., Duzhin~S.V., Khor'kova~N.G., Krasil'shchik~I.S.,
Samokhin~A.V., Torkhov~Yu.N., Verbovetsky~A.M. and Vinogradov~A.M.,
{\it Symmetries and Conservation Laws for Differential Equations of Mathematical Physics},
American Mathematical Society, Providence, RI, 1999.

\bibitem{Boyko&Popovych:Crum1955}
Crum M.M.,
Associated Sturm--Liouville systems,
{\it Quart. J. Math. Oxford Ser. (2)}, 1955, V.6, 121--127.

\bibitem{Boyko&Popovych:Kunzinger&Popovych2008}
Kunzinger M. and Popovych R.O.,
Potential conservation laws,
{\it J. Math. Phys.}, 2008, V.49, 103506, 34~pp., arXiv:0803.1156.

\bibitem{Boyko&Popovych:Matveev1979}
Matveev V.B.,
Darboux transformation and explicit solutions of the Kadomtcev--Pet\-viasch\-vily equation, depending on functional parameters,
{\it Lett. Math. Phys.}, 1979, V.3, 213--216.

\bibitem{Boyko&Popovych:Matveev&Salle1991}
Matveev V.B. and Salle M.A.,
{\it Darboux Transformations and Solitons},
Springer-Verlag, Berlin, 1991.

\bibitem{Boyko&Popovych:Marvan2004}
Marvan M., Reducibility of zero curvature representations with application to recursion operators,
{\it Acta Appl. Math.}, 2004, V.83, 39--68, arXiv:nlin.SI/0306006.

\bibitem{Boyko&Popovych:Olver1993}
Olver P.J.,
{\it Application of Lie Groups to Differential Equations},
Springer-Verlag, New York, 1993.

\bibitem{Boyko&Popovych:Popovych&Ivanova2005}
Popovych R.O. and Ivanova N.M.,
Hierarchy of conservation laws of diffusion-convection equations,
{\it J. Math. Phys.}, 2005, V.46, 043502, 22~pp., arXiv:math-ph/0407008.

\bibitem{Boyko&Popovych:Popovych&Kunzinger&Ivanova2008}
Popovych R.O., Kunzinger M. and Ivanova N.M.,
Conservation laws and potential symmetries of linear parabolic equations,
{\it Acta. Appl. Math.},  2008, V.100, 113--185, arXiv:0706.0443.

\bibitem{Boyko&Popovych:Popovych&Sergyeyev2010}
Popovych R.O. and Sergyeyev A.,
Conservation laws and normal forms of evolution equations,
{\it Phys. Lett. A}, 2010, V.374, 2210--2217, arXiv:1003.1648.

\bibitem{Boyko&Popovych:Sergyeyev2000}
Sergyeyev~A.,
On recursion operators and nonlocal symmetries of evolution equations, in
{\em Proceedings of the Seminar on Differential Geometry},
Silesian University in Opava, Opava, 2000, 159--173, arXiv:nlin.SI/0012011.

\bibitem{Boyko&Popovych:Wahlquist&Estabrook1975}
Wahlquist~H.D. and Estabrook~F.B.,
Prolongation structures of nonlinear evolution equations,
{\it J. Math. Phys.}, 1975, V.16, 1--7.

\bibitem{Boyko&Popovych:Wolf2002}
Wolf~T.,
A comparison of four approaches to the calculation of conservation laws,
{\it Eur. J. Appl. Math.}, 2002, V.13, Part 5, 129--152, arXiv:cs.SC/0301027.

\end{thebibliography}
\end{document}